\documentclass[conference,letterpaper]{IEEEtran}

\usepackage[utf8]{inputenc} 
\usepackage[T1]{fontenc}
\usepackage{url}
\usepackage{ifthen}
\usepackage{cite}
\usepackage[cmex10]{amsmath} 

\usepackage{mathtools}
\usepackage{float}
\usepackage{array}
\usepackage{tikz}
\usepackage{algpseudocode}
\usepackage{algorithm}
\usepackage{amssymb}
\usepackage{amsthm}

\usepackage{caption}
\usepackage{subcaption}
\usepackage{setspace}
\usepackage[font=small,labelfont=bf]{caption}

\usepackage{tikz-qtree}
\usepackage{tikz}

\usepackage[normalem]{ulem}

\usetikzlibrary{shapes.geometric}
\newfloat{procedure}{htbp}{loa}
\floatname{procedure}{Procedure} 
\usepackage{enumitem}
\usepackage[square,numbers]{natbib}
\usepackage{soul}
\usepackage{float}
\usepackage{subfig}
\usepackage{caption}
\usepackage{capt-of}
\usepackage{xcolor}
\usepackage{changes}
\usepackage{breqn}

\newtheorem{theorem}{Theorem}

\newtheorem{corollary}{Corollary}

\newtheorem{definition}{Definition}

\newtheorem{lemma}{Lemma}

\newtheorem{proposition}{Proposition}

\newcommand{\R}{\mathbb{R}} 
\newcommand{\N}{\mathbb{N}}
\newcommand{\x}{\mathbf x}
\newcommand{\bu}{\textbf{u}}
\newcommand{\bv}{\textbf{v}}
\newcommand{\vecp}{\textbf{p}}
\newcommand{\T}{\mathcal{T}}

\newcommand{\mbbS}{{\mathbb S}}

\newcommand{\bfS}{\mathbf S}

\newcommand{\MCP}{\mathbb{H}}

\newcommand{\g}{g} 
\newcommand{\States}{\mathbb{S}}
\newcommand{\cost}{\mathcal{\ell}}

\newcommand{\lp}{\left(}
\newcommand{\rp}{\right)}
\newcommand{\la}{\left\{}
\newcommand{\ra}{\right\}}

\begin{document}
\title{Better Algorithms for Constructing Minimum Cost Markov Chains and AIFV Codes} 


\author{%
  \IEEEauthorblockN{Mordecai J.~Golin}
  \IEEEauthorblockA{University of Massachusetts, Amherst\\
                    Email: mgolin@umass.edu}
  \and
  \IEEEauthorblockN{Reza Hosseini Dolatabadi and Arian Zamani}
  \IEEEauthorblockA{Department of Computer Engineering\\ 
                    Sharif University of Technology, Tehran, Iran\\
                    Emails: \{reza.dolatabadi256\}, \{arian.zamani243\}@sharif.edu}

}


\maketitle


\begin{abstract}
Almost Instantaneous Fixed to Variable (AIFV) coding is  a relatively new method of lossless coding that, unlike Huffman  coding,  uses more than one coding tree.  
The problem of constructing optimal AIFV codes is 
a special case of that of constructing minimum cost Markov Chains.  
This paper provides the first complete proof of correctness for the previously known iterative algorithm for constructing such Markov chains.  

A recent  work  describes how to efficiently solve the Markov Chain problem by first constructing  a {\em Markov Chain Polytope} and then running the 
 Ellipsoid algorithm for linear programming on it.
This  paper's second result is that,  in the  AIFV case, a special property of the polytope   instead permits  solving the corresponding linear program using simple binary search. 
\end{abstract}

\section{Introduction}\label{sec: intro}

In what follows $m,n$ are positive integers, $[m]=\{0,1\,\ldots,m-1\}$ and $\langle n \rangle = \la1,2,\dots,n\ra$.

Consider  a stationary memoryless source with alphabet   $\Sigma=\{\sigma_1, \sigma_2, \dots, \sigma_n\}$ in which symbol
$\sigma_i$ is generated with  probability $p_i$. Binary Huffman codes encode each $\sigma_i$ as a binary string. Huffman codes are usually represented as a tree, with the codewords being leaves of the tree.

$m$-ary Almost Instantaneous Fixed to Variable  (AIFV-$m$) coding   
\cite{iterative_3,iterative_4,aifv_mr,aifv_m,dp_2,iterative_m,iterative_2}
is a relatively new method of lossless coding that encodes using $m$ different coding trees (switching between the trees using a complicated rule).  Minimum cost AIFV-$m$ codes are interesting because they can approach Shannon Entropy closer than Huffman coding. 

The initial method for constructing minimum cost AIFV-$m$ codes was an iterative ("fastest descent") one developed 
in   \cite{iterative_3,iterative_4,iterative_m,iterative_2}.  As noted in \cite{iterative_4}, this algorithm actually solved the more general problem of constructing a minimum-cost $m$-state Markov chain. This algorithm was therefore also able to be used to solve the problem of finding better parsing dictionaries using multiple parse trees
\cite{iwata2021aivf}  as well as better lossless codes for finite channel coding  \cite{iwata2022joint}, two other problems that fit into the 
minimum-cost  Markov Chain framework.

The sequel is framed in terms of the general Markov Chain problem and not the specific  AIFV one. (For completeness, the appendix provides a complete description of the motivating  AIFV problem.)

The iterative algorithm referenced above
ran in 
exponential time. 
\cite{golin2020polynomial} described a binary-search algorithm running in $O(n^3 b)$ 
time for solving the AIFV-$2$ coding problem; this corresponded to constructing  a 
minimun cost $2$-state Markov chain. Recall that $n$ is the size of the source alphabet.  
$b$ was the maximum number of bits required to 
represent 
any of the $p_i.$
 \cite{golinaifv-m} recently showed how to transform the general $m$-state Markov chain problem into a problem on  corresponding {\em Markov Chain Polytope (MCP)}
in which the highest point on the MCP corresponds somehow to a minimum-cost Markov Chain. This transforms the problem into a linear programming one 
which, even though the MCP has exponential size,  can be efficiently solved using the Ellipsoid algorithm \cite{ellipsoid}. If the MCP has a polynomial time separation oracle, which the AIFV MCP does, the problem could then be solved in polynomial time.  A major weakness of this last result is that it is purely theoretical, since the Ellipsoid algorithm is not very efficient in practice.

This paper contains two separate results.

The first proves the correctness of the iterative algorithm given in \cite{iterative_3,iterative_4,iterative_m,iterative_2} and rewritten as Algorithm \ref {alg: iterative} below. The existing proofs only showed that if the algorithm terminates, it terminates at a correct  (min-cost) solution.  That still left open the possibility that the algorithm might infinitely loop.
We show that the algorithm always terminates for any instance of a Minimum Cost Markov Chain problem.

The second result is a much simpler algorithm for solving the AIFV-$m$ problem.  As intermediate steps in its analysis,  \cite{golinaifv-m} derived an additional property of the specific MCP associated with the AIFV-$m$ problem (Lemma  \ref {lemma:poncare miranda} below).  We show that this implies that linear programming on this polytope can be solved using easily implementable simple binary search rather than requiring the more complicated Ellipsoid algorithm. 

\subsection{The Minimum Cost Markov Chain Problem} \label{sec: Minimum Cost Markov Chain Problem}

Recall that a Markov Chain $ \bfS =(S_0,\ldots,S_{m-1})$  is 
specified by knowing the {\em transition probabilities}
$\{q_j(S_k)\}_{j \in [m]}$ for each 
$k \in [m].$ 
These satisfy $\forall {j \in [m]},\, q_j(S_k) \ge 0$ and 
$\sum_{j \in [m]} q_j(S_k)=1.$  If the Markov chain is {\em ergodic},
$\bfS$ has a unique stationary distribution 
$\pi(\bfS) =(\pi_0(\bfS),\ldots,\pi_{m-1}(\bfS))$.

Our problem is on Markov Chains with {\em rewards}.  That is, each state $S_k$ has a real nonnegative {\em reward} or {\em cost} $\cost (S_k).$

If $\pi(\bfS)$ exists, e.g., if 
$\bfS$
is ergodic then its {\em gain} 
\cite{gallager2011discrete}  or (average steady state) cost is 
$\mbox{\rm cost}(\bfS)=\sum_{k\in [m]} \cost(S_k) \cdot \pi_k(\bfS).$

Our problem is to find the chain $\bfS$ with minimum cost among a specifically defined set of permissible Markov Chains.  More explicitly:
\begin{definition}   Fix 
$m>1.$ 
\begin{enumerate}
%
\item $\forall k \in [m],\,$ let $\mbbS_k$ be some  finite given set of ``permissible state $k$'s'', satisfying that $\forall  S_k \in \mbbS_k,\,  q_0(S_k) >0.$  
\item A Markov Chain $ \bfS =(S_0,\ldots,S_{m-1})$ is  {\em permissible}  if $\forall k \in [m],\,  S_k \in \mbbS_k.$   
\item the set of permissible Markov chains is
$\mbbS = \bigtimes_{k=0}^{m-1} \mbbS_k= 
\left\{
(S_0,\ldots,S_{m-1}) \mid \forall k \in [m],\, S_k \in \mbbS_k
\right\}$ 
\end{enumerate}
The Minimum Cost Markov Chain (MCMC)  Problem is to find a $\bfS \in \mbbS$ satisfying
$\mbox{\rm cost}(\bfS) = \min_{\bfS'\in \mbbS}\mbox{\rm cost}(\bfS').$
\end{definition}

The specific structure of the $\mbbS_k$ is different from problem to problem. 
For intuition, 
in the original motivating example of AIFV-$m$ coding, AIFV-$m$ codes were $m$-tuples of coding trees.  The $k$'th tree in such a code could be any permissible type-$k$ tree (independent of the other trees in the tuple).  The $q_j(S_k)$ were the probabilities of encoding a source symbol with a type-$j$ tree immediately after encoding the previous source symbol with a specific type-$k$ tree. The average number of bits used to encode a source symbol turned out to be exactly the cost of the associated Markov Chain.  For more information on this correspondence please see the appendix.

We also note that condition (1) in the definition implies that 
 $\bfS=(S_0,\ldots,S_{m-1})$ is an ergodic {\em unichain},  with one aperiodic recurrent class  (containing $S_0$) and, possibly, some transient states. $\bfS$ therefore has a unique stationary distribution 
$\pi(\bfS)$ so $\mbox{\rm cost}(\bfS)$ is well defined for all 
$\bfS \in \mbbS.$

\subsection{The Markov Chain Polytope}
 The next set of definitions are from  \cite{golinaifv-m}
 which  showed how to transform the 
 MCMC 
 problem into a linear programming one by first transforming permissible states into $m$-dimensional hyperplanes and then working with the polytope defined by their lower envelope.



\begin{definition} \label{hp def} Let $k \in [m].$ In what follows, $\x$ will always satisfy 
        $\x = \left[x_1,x_2\dots,x_{m-1}\right]^T \in \mathbb{R}^{m-1}$.

    \begin{itemize}
        \item Define $f_k: \mathbb{R}^{m-1} \times \States_k \rightarrow \mathbb{R}$ as follows:
        \[
            f_0 \lp \x,S_0 \rp  = \cost \lp S_0 \rp  + \sum_{j=1}^{m-1}q_j \lp S_0 \rp x_j
        \]
        
        \[
            \forall k > 0: f_k \lp \x,S_k \rp  = \cost \lp S_k \rp  + \sum_{j=1}^{m-1}q_j \lp S_k \rp x_j - x_k
        \]
        \item Define $\g_k : \mathbb{R}^{m-1} \rightarrow \mathbb{R}$ as follows:

        \[
            g_k \lp \x \rp  = \min_{S_k \in \States_k} f_k \lp \x,S_k \rp 
        \]
        \[
            S_k \lp \x \rp  = \arg\min_{S_k \in \States_k} f_k \lp \x,S_k \rp 
        \]
        \[
            \bfS \lp \x \rp  = \bigl( S_0 \lp \x \rp , \ldots, S_{m-1} \lp \x \rp  \bigr)
        \]
        We call $g_k(x)$ the \textit{lower-envelope of type $k$}.
        \item Define $h: \mathbb{R}^{m-1} \rightarrow \mathbb{R}$ as 
        \[
            h \lp \x \rp  = \min_{k} g_k \lp \x \rp 
        \]
The {\em Markov Chain Polytope} corresponding to $\mbbS$ is 
$$ \MCP = \left\{(\x,y) \in \mathbb{R}^{m} \mid  0 \le y \le h(\x)\right\}.$$
The {\em height} of $\MCP$ is ${\rm height}(\MCP)=\max_{(\x,y)\in \MCP} y.$
%
 %
    \end{itemize}

\end{definition}

\cite{golinaifv-m}
 described the relationship of this polytope to the 
 MCMC
 problem.

\begin{proposition}[Lemma 3.1 in ~\cite{golinaifv-m}] \label{prop: big paper}
Let $ \bfS =(S_0,\ldots,S_{m-1})$ be any permissible   Markov chain   and 
     $f_k \lp \x,S_i \rp$, $k \in [m],$ its  associated  hyperplanes. Then these $m$ hyperplanes intersect  at a unique point 
    $(\x,y) \in \mathbb{R}^{m}$. Furthermore,  $ y \ge {\rm height}(\MCP).$
 
    Such a point  $(\x,y)$  will be called the multi-typed intersection corresponding to $S$. 
 %
 %
\end{proposition}
Note that the proposition implies that if $(\x,y)$  is the  multi-typed intersection corresponding to $ \bfS =(S_0,\ldots,S_{m-1}),$
 $${\rm cost}(\bfS)= f_0(\x,S_0)=f_1(\x,S_1)=\cdots=f_{m-1}(\x,S_{m-1}).$$
It  also  immediately implies that the lowest ($y$-coordinate) multi-typed intersection point corresponds to the cheapest Markov chain and its height is at least 
$\rm {height}(\MCP)$. Thus
\begin{corollary} \label{cor : highest = cheapest}
   If  some $(\x,y) \in \MCP$ is a multi-typed intersection point $(\x,y)$, then 
 $(\x,y)$ is 
    a highest point on 
 $\MCP,$ i.e., $ y = {\rm height}(\MCP).$
Furthermore,  $\bfS(\x)$ is 
    a  min cost Markov chain. 
\end{corollary}

According to the above statements, it is sufficient to prove that $H$ contains a multi-typed intersection point. Then finding the cheapest Markov chain will be equivalent to finding the highest point of the polytope $H$.  We will use this fact later.

\section{The Iterative Algorithm}\label{sec: iterative algorithm}

\cite{iterative_3,iterative_4,iterative_m,iterative_2} introduced the following iterative algorithm for solving the AIFV-$m$ problem. As noted by \cite{iterative_4}, it actually solves the generic MCMC Problem.


\begin{algorithm}[H]
\caption{Iterative Algorithm to Find The Optimal $\T \in \States$}\label{alg:dp}
\label{alg: iterative}
\begin{algorithmic}[1]
\Require $\States$ and $\bfS(\x)$ defined in Section \ref{sec: intro}
\algrenewcommand\algorithmicrequire{\textbf{Initialization:}}
\Require Let $\vecp_0 \in \mathbb{R}^{m-1}$ be the projection of any  arbitrary multi-typed intersection point.
\While{$\vecp_i \neq \vecp_{i-1}$ }
\State $i \leftarrow i+1$
\State  $\vecp_{i} \leftarrow $ the projection of the multi-typed intersection point  corresponding to $\bfS(\vecp_{i-1}).$ 
\label{algline: intersection part}    
\EndWhile \\
\Return $\bfS(\vecp_i)$
\end{algorithmic}
\end{algorithm}
In the algorithm, 
the  projection of 
$\x' = [x_1,x_2,\dots,x_m]^T \in \R^{m}$ on $\R^{m-1}$ is 
$\x = [x_1,x_2,\dots,x_{m-1}]^T \in \R^{m-1}$.

The algorithm starts by choosing any arbitrary $\bfS =(S_0,\ldots,S_{m-1})\in \mbbS,$  setting $\x'_0=(\x_0,y_0)$ to be the 
multi-typed intersection point corresponding to $S$ and the starting point as $\vecp_0=\x_0$, i.e., the projection of
$\x'_0$  on $\R^{m-1}$. 


Note that the algorithm strongly depends upon Proposition \ref{prop: big paper} to ensure that a unique multi-typed intersection point corresponding to any $\bfS \in \mbbS$ always exists.


The discussions in \cite{iterative_3,iterative_4,iterative_m,iterative_2} don't actually prove that the algorithm is always correct. 
They rather prove that if 
it terminates,
then it outputs a correct solution.  
They did not prove though, that 
it always terminates.
It could end up looping infinitely.  Our first new result will be a proof of termination and correctness in all cases.



The next section proves the correctness of  Algorithm \ref{alg: iterative}.
For completeness,  we note that  Algorithm \ref{alg: iterative} as presented is actually not  the iterative algorithm 
presented in \cite{iterative_3,iterative_4,iterative_m,iterative_2}.  But, it is shown in 
\cite{golinaifv-m} that Algorithm \ref{alg: iterative}  is identical to the earlier one(s) after a change of variables.  This change of variables vastly simplifies the proof of correctness

\section{Proof of Correctness of Algorithm \ref{alg: iterative}}\label{sec: statements}
Let 
\[
\textbf{u} = \left[ u_1,u_2,\dots,u_{n} \right]^T, \textbf{v} = \left[ v_1,v_2,\dots,v_{n} \right]^T \in \R^{n}
\]
be two vectors. Define $\textbf{u} \preceq \textbf{v}$ if
$
    \forall \, i \in \langle n \rangle: \; u_i \leq v_i.
    $
The proof requires partitioning   $\mathbb{R}^{m-1}$ into $m$ cones.
\begin{definition}
    Let $k \in [m].$ Define:
       \[
        C_0 = \la \textbf{u}  \in \mathbb{R}^{m-1} \mid \textbf{u} \preceq  0  \ra,
    \]
    \begin{align}
        \forall k > 0 : C_k = \Big\{ \textbf{u} \in 
        \mathbb{R}^{m-1} \mid u_k > 0 \;\; \text{and} \;\; \forall j: u_k \geq u_j \Big\}. \nonumber
    \end{align}
\end{definition}

\begin{lemma} \label{lem : main lemma}
       Fix $i \in [m]$ and let  $S_i \in \bfS_i.$
       Let  $\bu,\bv \in \mathbb{R}^{m-1}$ with $\bu\not=\bv$   and
       $\bv-\bu \in C_i.$ Then
%
    \[
        f_i \lp \bv,S_i \rp  \leq f_i \lp \bu,S_i \rp 
    \]
    Moreover, if $i \neq 0$,  
        $f_i \lp \bv,S_i \rp  < f_i \lp \bu,S_i \rp.$
    
\end{lemma}
\begin{proof}
We first analyze the case $i=0,$ i.e., 
 $\bv-\bu \in C_0$,  or, equivalently, $\bv \preceq  \bu$. By Definition \ref{hp def}:
\[
    f_0 \lp \bu,S_0 \rp  = \cost \lp S_0 \rp  + \sum_{j=1}^{m-1}q_j \lp S_0 \rp u_j
\]
\[
    f_0 \lp \bv,S_0 \rp  = \cost \lp S_0 \rp  + \sum_{j=1}^{m-1}q_j \lp S_0 \rp v_j
\]
Therefore,
\begin{equation} \label{eq:1}
    f_0 \lp \bv,S_0 \rp  - f_0 \lp \bu,S_0 \rp  = \sum_{j=1}^{m-1}q_j \lp S_0 \rp  \lp v_j - u_j \rp  
\end{equation}

By assumption, $\bv \preceq  \bu$, so for all $1 \leq j < m$,  $v_j - u_j \leq 0$. Hence, by Equation (\ref{eq:1}),
\[
    f_0 \lp \bv,S_0 \rp  - f_0 \lp \bu,S_0 \rp  \leq 0.
\]
We now analyze the cases in which $i \not=0,$ i.e., 
 $\bv-\bu \in C_i$ for  $i \in  \langle m-1 \rangle$. Define the set of indices $I \subseteq \langle m-1 \rangle$ as 
\begin{equation}\label{eq:2}
    I =  \la j \in [m-1] \mid v_j > u_j  \ra.
\end{equation}
By assumption $i \in I$,  so $I \neq \varnothing$. By Definition \ref{hp def},
$$ \small   f_i \lp \bu,S_i \rp  - f_i \lp \bv,S_i \rp  
    =\sum_{j \neq i} q_j \lp S_i \rp  \lp u_j-v_j \rp \nonumber + \lp q_i \lp S_i \rp  - 1 \rp  \lp u_i - v_i \rp.
$$

From Equation  (\ref{eq:2}),
\begin{align}\label{eq:4}
    f_i \lp \bu,S_i \rp  - f_i \lp \bv,S_i \rp  &= \sum_{j \in I \setminus \la i \ra } q_j \lp S_i \rp  \lp u_j-v_j \rp \nonumber  \\
    &+ \sum_{k \notin I} q_k \lp S_i \rp  \lp u_k-v_k \rp  \; \nonumber \\
    &+ \;  \lp q_i \lp S_i \rp  - 1 \rp  \lp u_i - v_i \rp
\end{align}

By the definition of $C_i,$
\begin{equation} \label{eq:5}
    \sum_{j \in I \setminus \la i \ra } q_j \lp S_i \rp  \lp u_j-v_j \rp  \geq  \lp \sum_{j \in I \setminus \la i \ra } q_j \lp S_i \rp  \rp  \lp u_i - v_i \rp 
\end{equation}
and by the definition of $I$
\begin{equation} \label{eq:6}
    \sum_{k \notin I} q_k \lp S_i \rp  \lp u_k-v_k \rp  \geq 0.
\end{equation}
Since 
\[
    1-q_i \lp S_i \rp  =  \lp \sum_{j \in \langle m-1 \rangle \setminus \la i\ra }q_j \lp S_i \rp   \rp + q_0 \lp S_i \rp
\]
by Equations (\ref{eq:4}), (\ref{eq:5}) and (\ref{eq:6}) we can conclude:
\begin{align*}
    f_i \lp \bu,S_i \rp  - f_i \lp \bv,S_i \rp  \geq  \lp \sum_{j \in I \setminus \la i \ra } q_j \lp S_i \rp  \rp  \lp u_i - v_i \rp  + 0 \\
    -  \lp  \lp \sum_{j \in \langle m-1 \rangle  \setminus \{i\} }q_j \lp S_i \rp   \rp + q_0 \lp S_i \rp  \rp  \lp u_i - v_i \rp.
\end{align*}
Hence
\begin{align*}
    f_i \lp \bu,S_i \rp  &- f_i \lp \bv,S_i \rp  \\
    &\geq \lp  \lp \sum_{j \in \langle m-1 \rangle  \setminus I }q_j \lp S_i \rp   \rp + q_0 \lp S_i \rp  \rp  \lp v_i - u_i \rp\\
   &\geq q_0 \lp S_i \rp  \lp v_i - u_i \rp  > 0
\end{align*}
completing the proof.
\end{proof}

\begin{theorem} \label{prop: convergence proof}
     Algorithm \ref{alg: iterative} always terminates in finite time. 
     Furthermore, at termination, $\bfS(\vecp_i)$ is a minimum cost Markov Chain and the multi-typed intersection point associated with $\bfS(\vecp_i)$ is in $\MCP.$
\end{theorem}
\begin{proof}
 Since the set $\mbbS$ of permissible Markov Chains is finite 
      the set $ \la \vecp_0, \vecp_1, \dots  \ra $ is finite. 
      
      Now set $c_0=y_0$. For  $i >0,$ define $c_i= \mbox{\rm cost}\lp \bfS\lp \vecp_{i-1}\rp \rp.$
 From Proposition \ref{prop: big paper} and the definition of the $\vecp_i$
    \begin{align*} \label{eq : def of costs}
        c_i &= f_0 \lp \vecp_i,S_0 \lp \vecp_{i-1} \rp  \rp  = f_1 \lp \vecp_i,S_1 \lp \vecp_{i-1} \rp  \rp  \\
        &= \dots = f_{m-1} \lp \vecp_i,S_{m-1} \lp \vecp_{i-1} \rp  \rp 
    \end{align*}
    
        We claim that for all $i > 1$, (i) $c_i \leq c_{i-1}$. Furthermore, (ii) if $c_i = c_{i-1}$,  $\vecp_i \preceq \vecp_{i-1}$.

    The $C_j$ partition   $\mathbb{R}^{m-1}$ so there exists some $j$ such that 
        $\vecp_i - \vecp_{i-1} \in C_j$. From  Lemma \ref{lem : main lemma},   
        \begin{equation} \label{eq: cone inequality}
            f_j \lp \vecp_{i-1},S_j \lp \vecp_{i-1} \rp  \rp  \geq f_j \lp \vecp_i,S_j \lp \vecp_{i-1} \rp  \rp  = c_i.
        \end{equation}
        
        From the definition of $S_i(x)$ (via $\arg\min$)  in Definition \ref{hp def}, 
        \[
            c_{i-1} =  f_j \lp \vecp_{i-1},S_j \lp \vecp_{i-2} \rp  \rp \geq  f_j \lp \vecp_{i-1},S_j \lp \vecp_{i-1} \rp  \rp.
        \]
        Combining with  inequality (\ref{eq: cone inequality}), yields (i),
        \begin{equation} \label{eq: claimed ineq}
           \forall i > 0,\quad  c_{i} \leq c_{i-1}.
        \end{equation}
        Moreover, by  Lemma \ref{lem : main lemma} we know if $j > 0$, the inequality in (\ref{eq: cone inequality}) is strict; thus, inequality (\ref{eq: claimed ineq}) is also strict. So, if $c_i = c_{i-1}$, then   $j=0$ so  $\vecp_i \preceq \vecp_{i-1}$ and (ii) is proved.
        

Since  the set $ \la \vecp_0, \vecp_1, \dots  \ra $ is finite,  $ \la c_0, c_1, \dots  \ra $ is also finite.
Inequality (\ref{eq: claimed ineq}) thus implies that the $c_i$ converge in finite time, i.e., there exists $N \in \N$ s.t. for all $i > N$,  $c_{i}=c_{i-1}.$

 Claim (ii) then implies that for all $i > N$,  $\vecp_i \preceq \vecp_{i-1}$. The finiteness of the set $ \la \vecp_0, \vecp_1, \dots  \ra $,  then immediately implies that,  for all sufficiently large $n \in \N$, $\vecp_n = \vecp_{n+1}$, i.e., the algorithm terminates.




   Note that if $\vecp_n = \vecp_{n+1}$, then $c_n = c_{n+1}$ so,
   $\forall k \in [m],$
   $$ c_n = c_{n+1}= f_k\lp \vecp_{n+1},S_k\lp \vecp_{n}\rp \rp
   =f_k\lp \vecp_{n},S_k\lp \vecp_{n}\rp \rp=
   g_k\lp \vecp_{n}\rp.
   $$
   Thus $ c_n = h \lp \vecp_n \rp$, so 
   $(\vecp_n,c_n)\in \MCP$ is a multi-typed intersection point. Corollary \ref{cor : highest = cheapest} then implies that  $\bfS(\vecp_n)$ is a minimum cost Markov Chain.
\end{proof}

Lemma \ref{lem : main lemma}
also provides more information about the structure of $\MCP.$ This will be useful later.

\begin{theorem} \label{prop: main proposition for uniqueness}
     The multi-typed intersection point found by Algorithm \ref {prop: convergence proof} is the unique multi-typed intersection point in  $\MCP.$ Furthermore, it is a vertex of $\MCP.$
\end{theorem}
\begin{proof}
%
     To prove the uniqueness, assume $Q = (\x ,h)$ and $Q'=(\x',h)$  ($Q,Q' \in \R^m$ and $\x,\x' \in \R^{m-1}$)   are multi-typed intersection  points in $\MCP$. Note that from
     Corollary \ref{cor : highest = cheapest}
     $h = h' = h(\x) = h(\x')=\mbox{\rm height}\lp \MCP \rp.$
    
     It is sufficient to show that $Q = Q'$. Assume $Q \preceq Q'$ doesn't hold. By  Lemma \ref{lem : main lemma}, there exists an index $i\in<m-1>$ s.t.
    \begin{equation}\label{eq: uniqueness ineq}
       \forall S_i\in \mbbS_i,\ \quad  f_i \lp \x,S_i \rp  > f_i \lp \x',S_i \rp.
    \end{equation}
   Let  $f_i \lp \bar \x,S_i \rp $ denote the hyperplane associated with state $S_i.$ 
   Let $S_i\in \mbbS_i$ be such that 
   $Q$ lies on 
   $f_i \lp \bar \x,S_i \rp.$ 
    Then
    \[
        f_i \lp \x,S_i \rp  = h = h \lp \x' \rp  \leq f_i \lp \x',S_i \rp,
    \]
    contradicting  (\ref{eq: uniqueness ineq}). Hence
 $
        Q \preceq Q'.
        $
    We can symmetrically prove that  $Q' \preceq Q$. Thus, $Q = Q'$, and the uniqueness follows.

    If $Q$ is not a vertex of $\MCP,$ since it is  a highest point of  $\MCP$, there exist two  
    other points highest points  $Q_1, Q_2$ on $\MCP$ s.t.
    \[
        Q = \frac{1}{2}\lp Q_1 + Q_2\rp
    \]
    It is clear that at least for one $i \in \{1,2\}$, $\x_i - \x \notin C_0$; thus, by Lemma \ref{lem : main lemma}, there exists $j \in \langle m-1 \rangle$ s.t. for all $S_j \in \States_j,$
        $
        f_j(\x_i, S_j) < f_j(\x, S_j)
        $
    and thus,
    $
        g_j(\x_i) < g_j(\x)
        $
    which is a contradiction. Thus $Q$ must be a vertex of $\MCP.$
\end{proof}

\begin{corollary}\label{cor:gcor}
There exists a unique $\x^* \in \mathbb{R}^m$ satisfying 
$$g_0(\x^*)=
g_1(\x^*)=
\ldots =
g_{m-1}(\x^*).
$$
Furthermore, $\bfS(\x^*)$ is a solution to the MCMC problem.
\end{corollary}
\begin{proof} In the appendix.
\end{proof}

\section{A New Weakly Polynomial Time Algorithm For AIFV-$3$-coding}

The previous section proved the correctness of Algorithm \ref {alg: iterative} in solving any 
MCMC
problem.  
The running time of the algorithm depends upon two factors.  The first is how quickly $\bfS(\x)$ can be calculated. This is very  dependent upon how the states $S_i$ are defined.  In the special case of  AIFV-$m$ coding of  $n$ source-code words,  $\bfS(\x)$  was first solved using integer linear  programming \cite{iterative_2}  requiring time exponential in $n.$

When  $\x \in [0,1]^{m-1}$ this was improved  to polynomial time by using dynamic programs.   More specifically,  for $m=2,$   \cite{dp_2} developed a 
$O(n^5)$ DP, improved to $O(n^3)$ by \cite{dp_2_speedup}; for $m >2,$  \cite{iterative_m}
 solved it in 
$O(n^{2m+1})$ time, improved to $O(n^{m+2})$ by \cite{golin2022speeding}.  

The second factor determining  running time is the number of iterative steps performed by the algorithm.  It might have to iterate through ALL permissible Markov Chains in which case, at least for AIFV coding, it requires exponential time.

\cite{golinaifv-m} recently showed how to transform a MCMC problem into an  equivalent Linear Programming one of finding a highest point in the associated Markov Chain Polytope $\MCP.$  Although $\MCP$ might be defined by an exponential number of constraints this linear program can be solved   using the Ellipsoid method.
In the specific case of AIFV-$m$ coding this yielded a  polynomial  (in $n$ and $b$) time algorithm.

Unfortunately, this result was purely theoretical, since the Ellipsoid method is  known to be 
extremely difficult to practically implement.

In this section we derive further special properties
of the Markov Chain Polytope for the specific case of the AIFV-$3$ problem. These will permit replacing the Ellipsoid method by simple binary search and yielding a $O(n^5 b^2)$ time algorithm.

Due to space considerations we restrict the discussion to the specific case of AIFV-$3$ coding.  We claim, though, that this technique works for all AIFV-$m$ coding, for $m > 3$ as well.



  Recall that  for each  $i \in \{0,1,2\}$ and every $S_i\in \mbbS_i$ there exists a type-$i$ hyperplane 
$f_i(\x,S_i)$ that ranges over all  $\x = (x_1,x_2) \in \mathbb{R}^{2}$. Let $\lambda\in\mathbb{R}$  be  fixed. 
For $j \in \{0,1\}$ let $h_{j,\lambda}$  denote the hyperplane in $\mathbb{R}^{3}$ defined by setting $x_j=\lambda$.  Now consider the intersection of type-$0$ and type-$2$ hyperplanes with $h_{1,\lambda}$:
\begin{eqnarray}
    f_0((\lambda,x_2), S_0) &=& \ell(S_0) + q_1(S_0)\cdot \lambda +q_2(S_0)\cdot x_2,\nonumber\\
    f_2((\lambda,x_2), S_2) &=& \ell(S_2) + q_1(S_2)\cdot \lambda +q_2(S_2)\cdot x_2 - x_2 \nonumber \\
    &=& \ell(S_2) + q_1(S_2)\cdot \lambda +\big(q_2(S_2)-1\big)\cdot x_2 \nonumber 
\end{eqnarray}
These are both lines in $x_2.$
The intersection of each 
$f_0$
hyperplane with $h_{1,\lambda}$ corresponds to a line with a non-negative slope; since 
$q_2(S_i)=1-q_0(S_i)-q_1(S_i)<1,$ 
because $q_0(S_i)$ is always positive,
the intersection of each 
$f_2$
hyperplane with $h_{1,\lambda}$ corresponds to a line with a negative slope. 
Thus the intersection of $g_0$ and with $h_{1,\lambda}$ is the lower-envelope of non-negative slope lines and the intersection of $g_2$ and with $h_{1,\lambda}$ is the lower-envelope of negative slope lines. 
Thus
\begin{lemma} For any fixed $\lambda\in\mathbb{R}$, 
    $g_0$ and $g_2$ intersect exactly once on the $h_{1,\lambda}$ hyperplane.
\end{lemma}
Note that this means that, for any $x_1 \in \mathbb{R}$ there exists a {\em unique} $x_2$ such that  $g_0(x_1,x_2)=g_2(x_1,x_2).$
This permits defining the {\em function}
 $E_1: [0,1] \rightarrow \mathbb{R}^3$ to be a function that takes $x_1\in[0,1]$ as an input and outputs the unique intersection point of $g_0$ and $g_2$ on $h_{1,x_1}$. 
\begin{lemma} \label{lemma:continuous}
    $E_1$  is a piecewise linear continuous function. 
\end{lemma}
\begin{proof}
Included in the appendix.
\end{proof}
Everything up to this point would be correct for any Minimum Cost Markov Chain problem. 
 We now use the following known facts that are specific to the AIFV coding case, i.e., when the states are AIFV coding trees.
\begin{lemma}[\cite{golinaifv-m}, Lemma 7.2] \label{lemma:poncare miranda}
Let the states $S_i$ be as defined by the AIFV coding trees. 
    Let $m$ be fixed, $n\geq 2^{m}-1$, $\x\in\mathbb{R}^{m-1}$ and $k\in\{1,2,...,m-1\}.$ Then
    \begin{itemize}
        \item If $x_k =0$, $g_0(\x)-g_k(\x) \leq 0.$
        \item If $x_k =1$, $g_k(\x)-g_0(\x) \leq 0.$
    \end{itemize}
\end{lemma}
This lets us prove
\begin{lemma} \label{lemma:intersection between 0,1}
    Let $x_1\in[0,1]$ and $E(x_1)=(x_1,x_2,y)$ be the unique intersection of $g_0$ and $g_2$ on the $h_{1,x_1}$ hyperplane. Then $x_2\in[0,1]$.
\end{lemma}
\begin{proof} 
    From Lemma \ref{lemma:poncare miranda}, $g_0\big(x_1,0\big)- g_2\big(x_1,0\big)\le 0$ and $g_0\big(x_1,1\big)- g_2\big(x_1,1\big)\ge 0$. Since 
    $g_0\big(x_1,x_2\big)- g_2\big(x_1,x_2\big)$ is continuous in $x_2$ for fixed $x_1,$
    the unique intersection of $g_0$ and $g_2$ with the $h_{1,x_1}$ hyperplane must occur when  $x_2\in[0,1]$.
\end{proof}
\begin{lemma} \label{lemma:2d-binarysearch}
    If $\lambda\in[0,1]$ can be written using at most $O(b)$
    bits when represented in binary,  then $E_1(\lambda)$ can be calculated in  $O(n^5\cdot b)$ time.
\end{lemma}
\begin{proof} We use  a halving method similar to one used in \cite{golin2020polynomial} for solving the AIFV-$2$ problem.
     
Note that, for any $(x_1,x_2) \in [0,1]^2$ the
DP
algorithm from 
\cite{golin2022speeding} mentioned earlier can calculate $g_0(x_1,x_2)$ and $g_2(x_1,x_2)$ in $O(n^5)$ time.  Thus 
it can determine
 in $O(n^5)$ time  whether or not $g_0\big(x_1,x_2\big)- g_2\big(x_1,x_2\big)\le 0$.

From Lemma \ref{lemma:intersection between 0,1}, finding $E_1(\lambda)$ is equivalent to finding the unique $x^*_2(\lambda) \in [0,1]$ satisfying $g_0(\lambda,x^*_2(\lambda))=g_2(\lambda,x^*_2(\lambda)).$

From Lemma \ref{lemma:poncare miranda}, $g_0\big(x_1,0\big)- g_2\big(x_1,0\big)\le 0$ and $g_0\big(x_1,1\big)- g_2\big(x_1,1\big)\ge 0$.  Set 
$\ell = 0$ and $r=1$.  We can then simply do a halving search, after $s$ steps always maintaining an interval $[\ell,r]$ where $x^*_2(\lambda) \in [\ell,r]$ satisfying 
$\ell = t 2^{-s}$ and $r = (t+1) 2^{-s}$ for some $t \in [2^s].$ This requires $O(n^5 s)$ time in total.

We now note that, again using an argument similar to the one used in  \cite{golin2020polynomial}, in the AIFV case we know that the  $q_j(S_i)$ and  $\cost(S_i)$
are all integer multiples of $2^{-b}.$

This lower bounds the absolute values of the  non-zero slopes  and  $y$-intercepts of all lines  $f_i((\lambda,x_2),S_i)$ to be $2^{-O(b)}.$
This in turn implies (see \cite{golin2020polynomial} for more details) that after $ O(b)$ steps the interval $[\ell,r]$ will be small enough that $(\lambda,x_2^*(\lambda))$ will lie on one of the four lines,  $f_i((\lambda,x_2),S_i^d)$ where  
$i \in \{0,2\},$
$d\in \{\ell,r\}$ 
and $S_i^d=S_i(\lambda,d).$ The $S_i^d$ can be found in $O(n^5)$ time using the DP  procedure.
Since there are only 4 lines, $x^*_2(\lambda)$  itself can be found in only a further $O(1)$ time.
\end{proof}
Set $E_1^{\prime}: [0,1] \rightarrow \mathbb{R}^2$ to be the  function that takes $\lambda\in[0,1]$ as input and outputs the first two elements of $E_1(\lambda)$.
That is if $E_1(\lambda)=(\lambda,x_2,y)$ then $E_1^{\prime}=(\lambda,x_2)$.

\begin{algorithm}[h] \label{algorithm:3d-iterative-binary-search}
\caption{Finds  $x_1^{\prime}$ satisfying $|x_1^{*}-x_1^{\prime}|\leq2^{-b^{\prime}}$.}\label{alg:iterative-binarysearch}
\begin{algorithmic}[1]
\Require $\big\{p_1,...,p_n\big\}$ 
and $b^{\prime}$.
\algrenewcommand\algorithmicrequire{\textbf{Initialization:}}
\Require $l,r \gets 0,1$ and $\epsilon_0=2^{-b^{\prime}}$.
\Repeat
\State mid $\leftarrow$ $\frac{l+r}{2}$ ;
\State Calculate $p=E_1^{\prime}\big(\mbox{mid}\big)$.
\State $e_0 = g_0(p)$;  $e_1$ = $g_1(p)$
\If{$e_0 < e_1$}
\State $l = \mbox{mid}$
\Else 
\State $r =\mbox{mid}$ 
\EndIf
\Until{$r-l=\epsilon_0$}\\
\Return $l$
\end{algorithmic}
\end{algorithm}
Let $\x^{*}=(x_1^{*},x_2^{*},y^{*})$ be the unique intersection of $g_0$, $g_1$ and $g_2$. We will now construct an algorithm that helps us find a point $\x^{\prime}=(x_1^{\prime},x_2^{\prime},y^{\prime})$ that is close to $\x^{*}$. More specifically, we will prove the following lemma:
\begin{lemma} \label{lemma: binary search algorithm}
    Let $b^{\prime}$ be an integer such that $b^{\prime}=O(b)$, Algorithm \ref{alg:iterative-binarysearch} can find 
    $x'_1$
    satisfying $|x_1^{*}-x_1^{\prime}|\leq2^{-b^{\prime}}$ in $O(n^5\cdot b^2)$ time. Furthermore we can modify this algorithm to find $|x_2^{*}-x_2^{\prime}|\leq2^{-b^{\prime}}$ in the same time complexity.
\end{lemma}
\begin{proof}
    Included in the appendix.
\end{proof}
Let $b^{\prime}=14\cdot b +18$ and  $T_0^{\prime}$, $T_1^{\prime}$ and $T_2^{\prime}$ be the trees corresponding to $g_0(x_1^{\prime},x_2^{\prime})$, $g_1(x_1^{\prime},x_2^{\prime})$ and $g_2(x_1^{\prime},x_2^{\prime})$ respectively. 
From Proposition \ref{prop: big paper},
the hyperplanes corresponding to $T_0^{\prime}$, $T_1^{\prime}$ and $T_2^{\prime}$ intersect at exactly one point. This means that if $\x^{*}$  also lies on all three of these hyperplanes, then their intersection is $\x^{*}$ and therefore, from  Corollary 
\ref{cor:gcor},
$\big(T_0^{\prime}, T_1^{\prime}, T_2^{\prime}\big)$ forms an optimal answer to the AIFV-$3$ problem  and we are done.
\begin{lemma} \label{lemma: x* lies on Ti'}
    $\x^{*}$ lies on $f_0(\x,T_0^{\prime})$, 
    $f_1(\x,T_1^{\prime})$ and 
    $f_2(\x,T_2^{\prime})$.
\end{lemma}
\begin{proof}
    Included in the appendix
\end{proof}

To recap, 
Lemma \ref{lemma: binary search algorithm} shows us how to find 
appropriate $(x'_1,x'_2)$  in $O(n^5 b^2)$ time. The associated
$T_0^{\prime}$, $T_1^{\prime}$, $T_2^{\prime}$ can be found in $O(n^5)$ further time using the known DP algorithm. Lemma \ref{lemma: x* lies on Ti'} then ensures that this is a correct solution to the AIFV-$3$ problem.
Although a polynomial time solution for the AIFV-$3$ problem was previously presented in \cite{golinaifv-m}, that required using the Ellipsoid algorithm while the approach here only needed simple binary searches.

\section*{Acknowledgment}
Work of all authors partially supported by RGC CERG Grant 16212021.

\IEEEtriggeratref{4}



\bibliography{refrence}
\newpage
\appendix
\begin{center}
    \textbf{ AIFV-$m$ code definition}
\end{center}
[Note: This 2 page subsection with its definitions and accompanying examples and  diagrams has been copied, with permission, from \cite{golinaifv-m}.]

A binary AIFV-m code will be a sequence $(T_0, T_1, \cdots , T_{m-1})$ of $m$ binary
code trees satisfying Definitions \ref{def:aifv node-types} and \ref{def: binary aifv-m codes} below. Each $T_i$ contains $n$
codewords. Unlike in Huffman codes, codewords can be internal nodes.\\
\begin{definition} \label{def:aifv node-types}
    (Node Types in a Binary AIFV-m Code \cite{aifv_m}). Figure \ref{fig:node_typesX}.
Edges in an AIFV-m code tree are labelled as 0-edges or 1-edges. If node $v$
is connected to its child node $u$ via a 0-edge (1-edge) then $u$ is $v$’s 0-child
(1-child). We will often identify a node interchangeably with its associated
(code)word. For example $0^
210$ is the node reached by following the edges
0, 0, 1, 0 down from the root. Following \cite{aifv_m}, the nodes in AIFV-m code
trees can be classified as being exactly one of 3 different types:
\begin{itemize}
    \item Complete Nodes. A complete node has two children: a 0-child and a
1-child. A complete node has no source symbol assigned to it.
\item  Intermediate Nodes. A intermediate node has no source symbol assigned to it and has exactly one child. A intermediate node with
a 0-child is called a intermediate-0 node; with a 1-child is called a
intermediate-1 node
\item Master Nodes. A master node has an assigned source symbol and at
most one child node. Master nodes have associated degrees:
\begin{itemize}
    \item[$-$] a master node of degree $k = 0$ is a leaf.
    \item[$-$] a master node v of degree $k \geq 1$ is connected to its unique child by
a 0-edge. Furthermore it has exactly $k$ consecutive intermediate 0 nodes as its direct descendants, i.e., $v0^
t$
for $0 < t \leq k$ are
intermediate-0 nodes while $v0^{
k+1}$ is not a intermediate-0 node.
\end{itemize}
\end{itemize}
\end{definition}
\begin{figure} [t] 
    \centering

    \tikzset{every picture/.style={line width=0.75pt}} 
    
    \begin{tikzpicture}[x=0.75pt,y=0.75pt,yscale=-1,xscale=1]
    
    \draw   (19.4,17.3) .. controls (19.4,9.95) and (25.35,4) .. (32.7,4) .. controls (40.05,4) and (46,9.95) .. (46,17.3) .. controls (46,24.65) and (40.05,30.6) .. (32.7,30.6) .. controls (25.35,30.6) and (19.4,24.65) .. (19.4,17.3) -- cycle ;
    \draw   (1.4,49.3) .. controls (1.4,41.95) and (7.35,36) .. (14.7,36) .. controls (22.05,36) and (28,41.95) .. (28,49.3) .. controls (28,56.65) and (22.05,62.6) .. (14.7,62.6) .. controls (7.35,62.6) and (1.4,56.65) .. (1.4,49.3) -- cycle ;
    \draw   (39.4,48.3) .. controls (39.4,40.95) and (45.35,35) .. (52.7,35) .. controls (60.05,35) and (66,40.95) .. (66,48.3) .. controls (66,55.65) and (60.05,61.6) .. (52.7,61.6) .. controls (45.35,61.6) and (39.4,55.65) .. (39.4,48.3) -- cycle ;
    \draw   (84.4,14.3) .. controls (84.4,6.95) and (90.35,1) .. (97.7,1) .. controls (105.05,1) and (111,6.95) .. (111,14.3) .. controls (111,21.65) and (105.05,27.6) .. (97.7,27.6) .. controls (90.35,27.6) and (84.4,21.65) .. (84.4,14.3) -- cycle ;
    \draw   (84.4,53.3) .. controls (84.4,45.95) and (90.35,40) .. (97.7,40) .. controls (105.05,40) and (111,45.95) .. (111,53.3) .. controls (111,60.65) and (105.05,66.6) .. (97.7,66.6) .. controls (90.35,66.6) and (84.4,60.65) .. (84.4,53.3) -- cycle ;
    \draw    (97.7,27.6) -- (97.7,40) ;
    \draw   (124.4,13.3) .. controls (124.4,5.95) and (130.35,0) .. (137.7,0) .. controls (145.05,0) and (151,5.95) .. (151,13.3) .. controls (151,20.65) and (145.05,26.6) .. (137.7,26.6) .. controls (130.35,26.6) and (124.4,20.65) .. (124.4,13.3) -- cycle ;
    \draw   (124.4,52.3) .. controls (124.4,44.95) and (130.35,39) .. (137.7,39) .. controls (145.05,39) and (151,44.95) .. (151,52.3) .. controls (151,59.65) and (145.05,65.6) .. (137.7,65.6) .. controls (130.35,65.6) and (124.4,59.65) .. (124.4,52.3) -- cycle ;
    \draw    (137.7,26.6) -- (137.7,39) ;
    \draw   (163.4,34.3) .. controls (163.4,26.95) and (169.35,21) .. (176.7,21) .. controls (184.05,21) and (190,26.95) .. (190,34.3) .. controls (190,41.65) and (184.05,47.6) .. (176.7,47.6) .. controls (169.35,47.6) and (163.4,41.65) .. (163.4,34.3) -- cycle ;
    \draw   (199.4,16.3) .. controls (199.4,8.95) and (205.35,3) .. (212.7,3) .. controls (220.05,3) and (226,8.95) .. (226,16.3) .. controls (226,23.65) and (220.05,29.6) .. (212.7,29.6) .. controls (205.35,29.6) and (199.4,23.65) .. (199.4,16.3) -- cycle ;
    \draw   (199.4,55.3) .. controls (199.4,47.95) and (205.35,42) .. (212.7,42) .. controls (220.05,42) and (226,47.95) .. (226,55.3) .. controls (226,62.65) and (220.05,68.6) .. (212.7,68.6) .. controls (205.35,68.6) and (199.4,62.65) .. (199.4,55.3) -- cycle ;
    \draw    (212.7,29.6) -- (212.7,42) ;
    \draw   (199.4,94.3) .. controls (199.4,86.95) and (205.35,81) .. (212.7,81) .. controls (220.05,81) and (226,86.95) .. (226,94.3) .. controls (226,101.65) and (220.05,107.6) .. (212.7,107.6) .. controls (205.35,107.6) and (199.4,101.65) .. (199.4,94.3) -- cycle ;
    \draw    (212.7,68.6) -- (212.7,81) ;
    \draw   (243.4,15.3) .. controls (243.4,7.95) and (249.35,2) .. (256.7,2) .. controls (264.05,2) and (270,7.95) .. (270,15.3) .. controls (270,22.65) and (264.05,28.6) .. (256.7,28.6) .. controls (249.35,28.6) and (243.4,22.65) .. (243.4,15.3) -- cycle ;
    \draw   (243.4,54.3) .. controls (243.4,46.95) and (249.35,41) .. (256.7,41) .. controls (264.05,41) and (270,46.95) .. (270,54.3) .. controls (270,61.65) and (264.05,67.6) .. (256.7,67.6) .. controls (249.35,67.6) and (243.4,61.65) .. (243.4,54.3) -- cycle ;
    \draw    (256.7,28.6) -- (256.7,41) ;
    \draw   (243.4,94.3) .. controls (243.4,86.95) and (249.35,81) .. (256.7,81) .. controls (264.05,81) and (270,86.95) .. (270,94.3) .. controls (270,101.65) and (264.05,107.6) .. (256.7,107.6) .. controls (249.35,107.6) and (243.4,101.65) .. (243.4,94.3) -- cycle ;
    \draw   (243.4,133.3) .. controls (243.4,125.95) and (249.35,120) .. (256.7,120) .. controls (264.05,120) and (270,125.95) .. (270,133.3) .. controls (270,140.65) and (264.05,146.6) .. (256.7,146.6) .. controls (249.35,146.6) and (243.4,140.65) .. (243.4,133.3) -- cycle ;
    \draw    (256.7,107.6) -- (256.7,120) ;
    \draw    (256.7,67.6) -- (256.7,81) ;
    \draw    (22,25.6) -- (14.7,36) ;
    \draw    (44,25.6) -- (52.7,35) ;
    
    \draw (26,9.4) node [anchor=north west][inner sep=0.75pt]    {$C$};
    \draw (91,6.4) node [anchor=north west][inner sep=0.75pt]    {$I_{0}$};
    \draw (132,5.4) node [anchor=north west][inner sep=0.75pt]    {$I_{1}$};
    \draw (165,26.4) node [anchor=north west][inner sep=0.75pt]    {$M_{0}$};
    \draw (201,6.4) node [anchor=north west][inner sep=0.75pt]    {$M_{1}$};
    \draw (206,47.4) node [anchor=north west][inner sep=0.75pt]    {$I_{0}$};
    \draw (200,85.4) node [anchor=north west][inner sep=0.75pt]    {$\neg I_{0}$};
    \draw (245,6.4) node [anchor=north west][inner sep=0.75pt]    {$M_{2}$};
    \draw (250,47.4) node [anchor=north west][inner sep=0.75pt]    {$I_{0}$};
    \draw (250,86.4) node [anchor=north west][inner sep=0.75pt]    {$I_{0}$};
    \draw (244,124.4) node [anchor=north west][inner sep=0.75pt]    {$\neg I_{0}$};
    \draw (7,20.4) node [anchor=north west][inner sep=0.75pt]    {$0$};
    \draw (83,26.4) node [anchor=north west][inner sep=0.75pt]    {$0$};
    \draw (200,27.4) node [anchor=north west][inner sep=0.75pt]    {$0$};
    \draw (198,67.4) node [anchor=north west][inner sep=0.75pt]    {$0$};
    \draw (241,27.4) node [anchor=north west][inner sep=0.75pt]    {$0$};
    \draw (240,66.4) node [anchor=north west][inner sep=0.75pt]    {$0$};
    \draw (241,106.4) node [anchor=north west][inner sep=0.75pt]    {$0$};
    \draw (53,15.7) node [anchor=north west][inner sep=0.75pt]    {$1$};
    \draw (123,26.4) node [anchor=north west][inner sep=0.75pt]    {$1$};

    \end{tikzpicture}

    \caption{Node types in a binary AIFV-$3$ code tree: complete node ($C$), intermediate-$0$ and intermediate-$1$ nodes ($I_0$, $I_1$), master nodes of degrees $0, 1, 2$ ($M_0$, $M_1$, $M_2$), ($M_0$ is a leaf)and non-intermediate-$0$ nodes ($\neg{I_0}$). The $ \neg {I_0}$ nodes can be complete, master, or intermediate-$1$ nodes, depending upon their location in the tree.}
    \label{fig:node_typesX}
\end{figure}
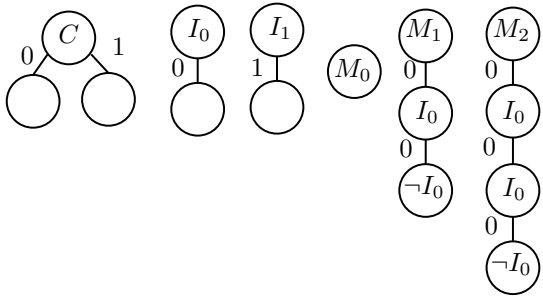

Binary AIFV-m codes are now defined as follows:
\begin{definition} \label{def: binary aifv-m codes}
    (Binary AIFV-m Codes \cite{aifv_m}). See Figure \ref{fig:example_aifv_3X}. Let $m \geq 2$ be
a positive integer. A binary AIFV-m code is an ordered m-tuple of m code
trees $(T_0, T_1, \cdots , T_{m-1})$ satisfying the following conditions:
\begin{enumerate}
    \item Every node in each code tree is either a complete node, a intermediate
node, or a master node of degree k where $0 \leq k < m$.
    \item For $k \geq 1$, the code tree $T_k$ has a intermediate-1 node connected to the root by exactly k 0-edges, i.e., the node $0^k$ is a intermediate-1 node.
\end{enumerate}

\begin{figure} [t] 
    \centering
    
    \tikzset{every picture/.style={line width=0.75pt}} 
    
    \begin{tikzpicture}[x=0.75pt,y=0.75pt,yscale=-1,xscale=1]
    
    \draw   (23,29.5) .. controls (23,25.91) and (25.91,23) .. (29.5,23) .. controls (33.09,23) and (36,25.91) .. (36,29.5) .. controls (36,33.09) and (33.09,36) .. (29.5,36) .. controls (25.91,36) and (23,33.09) .. (23,29.5) -- cycle ;
    \draw   (-1,56.5) .. controls (-1,49.6) and (4.6,44) .. (11.5,44) .. controls (18.4,44) and (24,49.6) .. (24,56.5) .. controls (24,63.4) and (18.4,69) .. (11.5,69) .. controls (4.6,69) and (-1,63.4) .. (-1,56.5) -- cycle ;
    \draw   (5,87.5) .. controls (5,83.91) and (7.91,81) .. (11.5,81) .. controls (15.09,81) and (18,83.91) .. (18,87.5) .. controls (18,91.09) and (15.09,94) .. (11.5,94) .. controls (7.91,94) and (5,91.09) .. (5,87.5) -- cycle ;
    \draw   (40,86.5) .. controls (40,82.91) and (42.91,80) .. (46.5,80) .. controls (50.09,80) and (53,82.91) .. (53,86.5) .. controls (53,90.09) and (50.09,93) .. (46.5,93) .. controls (42.91,93) and (40,90.09) .. (40,86.5) -- cycle ;
    \draw   (34,55.5) .. controls (34,48.6) and (39.6,43) .. (46.5,43) .. controls (53.4,43) and (59,48.6) .. (59,55.5) .. controls (59,62.4) and (53.4,68) .. (46.5,68) .. controls (39.6,68) and (34,62.4) .. (34,55.5) -- cycle ;
    \draw   (40,111.5) .. controls (40,107.91) and (42.91,105) .. (46.5,105) .. controls (50.09,105) and (53,107.91) .. (53,111.5) .. controls (53,115.09) and (50.09,118) .. (46.5,118) .. controls (42.91,118) and (40,115.09) .. (40,111.5) -- cycle ;
    \draw   (-1,119.5) .. controls (-1,112.6) and (4.6,107) .. (11.5,107) .. controls (18.4,107) and (24,112.6) .. (24,119.5) .. controls (24,126.4) and (18.4,132) .. (11.5,132) .. controls (4.6,132) and (-1,126.4) .. (-1,119.5) -- cycle ;
    \draw   (34,143.5) .. controls (34,136.6) and (39.6,131) .. (46.5,131) .. controls (53.4,131) and (59,136.6) .. (59,143.5) .. controls (59,150.4) and (53.4,156) .. (46.5,156) .. controls (39.6,156) and (34,150.4) .. (34,143.5) -- cycle ;
    \draw    (23,29.5) -- (11.5,44) ;
    \draw    (36,29.5) -- (46.5,43) ;
    \draw    (11.5,81) -- (11.5,69) ;
    \draw    (11.5,107) -- (11.5,94) ;
    \draw    (46.5,80) -- (46.5,68) ;
    \draw    (46.5,105) -- (46.5,93) ;
    \draw    (46.5,131) -- (46.5,118) ;
    \draw   (99,24.5) .. controls (99,20.91) and (101.91,18) .. (105.5,18) .. controls (109.09,18) and (112,20.91) .. (112,24.5) .. controls (112,28.09) and (109.09,31) .. (105.5,31) .. controls (101.91,31) and (99,28.09) .. (99,24.5) -- cycle ;
    \draw   (72,53.5) .. controls (72,49.91) and (74.91,47) .. (78.5,47) .. controls (82.09,47) and (85,49.91) .. (85,53.5) .. controls (85,57.09) and (82.09,60) .. (78.5,60) .. controls (74.91,60) and (72,57.09) .. (72,53.5) -- cycle ;
    \draw   (89,72.5) .. controls (89,68.91) and (91.91,66) .. (95.5,66) .. controls (99.09,66) and (102,68.91) .. (102,72.5) .. controls (102,76.09) and (99.09,79) .. (95.5,79) .. controls (91.91,79) and (89,76.09) .. (89,72.5) -- cycle ;
    \draw   (129,79.5) .. controls (129,75.91) and (131.91,73) .. (135.5,73) .. controls (139.09,73) and (142,75.91) .. (142,79.5) .. controls (142,83.09) and (139.09,86) .. (135.5,86) .. controls (131.91,86) and (129,83.09) .. (129,79.5) -- cycle ;
    \draw   (129,103.5) .. controls (129,99.91) and (131.91,97) .. (135.5,97) .. controls (139.09,97) and (142,99.91) .. (142,103.5) .. controls (142,107.09) and (139.09,110) .. (135.5,110) .. controls (131.91,110) and (129,107.09) .. (129,103.5) -- cycle ;
    \draw   (66,98.5) .. controls (66,91.6) and (71.6,86) .. (78.5,86) .. controls (85.4,86) and (91,91.6) .. (91,98.5) .. controls (91,105.4) and (85.4,111) .. (78.5,111) .. controls (71.6,111) and (66,105.4) .. (66,98.5) -- cycle ;
    \draw   (100,98.5) .. controls (100,91.6) and (105.6,86) .. (112.5,86) .. controls (119.4,86) and (125,91.6) .. (125,98.5) .. controls (125,105.4) and (119.4,111) .. (112.5,111) .. controls (105.6,111) and (100,105.4) .. (100,98.5) -- cycle ;
    \draw   (123,49.5) .. controls (123,42.6) and (128.6,37) .. (135.5,37) .. controls (142.4,37) and (148,42.6) .. (148,49.5) .. controls (148,56.4) and (142.4,62) .. (135.5,62) .. controls (128.6,62) and (123,56.4) .. (123,49.5) -- cycle ;
    \draw   (123,133.5) .. controls (123,126.6) and (128.6,121) .. (135.5,121) .. controls (142.4,121) and (148,126.6) .. (148,133.5) .. controls (148,140.4) and (142.4,146) .. (135.5,146) .. controls (128.6,146) and (123,140.4) .. (123,133.5) -- cycle ;
    \draw    (99,24.5) -- (78.5,47) ;
    \draw    (112,24.5) -- (135.5,37) ;
    \draw    (85,53.5) -- (95.5,66) ;
    \draw    (89,72.5) -- (78.5,86) ;
    \draw    (102,72.5) -- (112.5,86) ;
    \draw    (135.5,62) -- (135.5,73) ;
    \draw    (135.5,86) -- (135.5,97) ;
    \draw    (135.5,110) -- (135.5,121) ;
    \draw   (189,100.5) .. controls (189,96.91) and (191.91,94) .. (195.5,94) .. controls (199.09,94) and (202,96.91) .. (202,100.5) .. controls (202,104.09) and (199.09,107) .. (195.5,107) .. controls (191.91,107) and (189,104.09) .. (189,100.5) -- cycle ;
    \draw   (169,65) .. controls (169,61.41) and (171.91,58.5) .. (175.5,58.5) .. controls (179.09,58.5) and (182,61.41) .. (182,65) .. controls (182,68.59) and (179.09,71.5) .. (175.5,71.5) .. controls (171.91,71.5) and (169,68.59) .. (169,65) -- cycle ;
    \draw   (169,87.5) .. controls (169,83.91) and (171.91,81) .. (175.5,81) .. controls (179.09,81) and (182,83.91) .. (182,87.5) .. controls (182,91.09) and (179.09,94) .. (175.5,94) .. controls (171.91,94) and (169,91.09) .. (169,87.5) -- cycle ;
    \draw   (212,116.5) .. controls (212,112.91) and (214.91,110) .. (218.5,110) .. controls (222.09,110) and (225,112.91) .. (225,116.5) .. controls (225,120.09) and (222.09,123) .. (218.5,123) .. controls (214.91,123) and (212,120.09) .. (212,116.5) -- cycle ;
    \draw   (157,122.5) .. controls (157,115.6) and (162.6,110) .. (169.5,110) .. controls (176.4,110) and (182,115.6) .. (182,122.5) .. controls (182,129.4) and (176.4,135) .. (169.5,135) .. controls (162.6,135) and (157,129.4) .. (157,122.5) -- cycle ;
    \draw   (228,139.5) .. controls (228,132.6) and (233.6,127) .. (240.5,127) .. controls (247.4,127) and (253,132.6) .. (253,139.5) .. controls (253,146.4) and (247.4,152) .. (240.5,152) .. controls (233.6,152) and (228,146.4) .. (228,139.5) -- cycle ;
    \draw   (183,138.5) .. controls (183,131.6) and (188.6,126) .. (195.5,126) .. controls (202.4,126) and (208,131.6) .. (208,138.5) .. controls (208,145.4) and (202.4,151) .. (195.5,151) .. controls (188.6,151) and (183,145.4) .. (183,138.5) -- cycle ;
    \draw   (163,34.5) .. controls (163,27.6) and (168.6,22) .. (175.5,22) .. controls (182.4,22) and (188,27.6) .. (188,34.5) .. controls (188,41.4) and (182.4,47) .. (175.5,47) .. controls (168.6,47) and (163,41.4) .. (163,34.5) -- cycle ;
    \draw    (175.5,71.5) -- (175.5,81) ;
    \draw    (175.5,47) -- (175.5,58.5) ;
    \draw    (182,87.5) -- (195.5,94) ;
    \draw    (202,100.5) -- (218.5,110) ;
    \draw    (189,100.5) -- (169.5,110) ;
    \draw    (212,116.5) -- (195.5,126) ;
    \draw    (225,116.5) -- (240.5,127) ;
    
    \draw (6,50.4) node [anchor=north west][inner sep=0.75pt]    {$a$};
    \draw (42,48.4) node [anchor=north west][inner sep=0.75pt]    {$b$};
    \draw (6,115.4) node [anchor=north west][inner sep=0.75pt]    {$c$};
    \draw (40,137.4) node [anchor=north west][inner sep=0.75pt]    {$d$};
    \draw (131,44.4) node [anchor=north west][inner sep=0.75pt]    {$a$};
    \draw (73,91.4) node [anchor=north west][inner sep=0.75pt]    {$b$};
    \draw (108,92.4) node [anchor=north west][inner sep=0.75pt]    {$c$};
    \draw (130,125.4) node [anchor=north west][inner sep=0.75pt]    {$d$};
    \draw (170,28.4) node [anchor=north west][inner sep=0.75pt]    {$a$};
    \draw (164,116.4) node [anchor=north west][inner sep=0.75pt]    {$b$};
    \draw (190,133.4) node [anchor=north west][inner sep=0.75pt]    {$c$};
    \draw (235,132.4) node [anchor=north west][inner sep=0.75pt]    {$d$};
    \draw (6,27.4) node [anchor=north west][inner sep=0.75pt]    {$0$};
    \draw (-1,68.4) node [anchor=north west][inner sep=0.75pt]    {$0$};
    \draw (-1,94.4) node [anchor=north west][inner sep=0.75pt]    {$0$};
    \draw (32,67.4) node [anchor=north west][inner sep=0.75pt]    {$0$};
    \draw (33,91.4) node [anchor=north west][inner sep=0.75pt]    {$0$};
    \draw (34,118.4) node [anchor=north west][inner sep=0.75pt]    {$0$};
    \draw (77,24.4) node [anchor=north west][inner sep=0.75pt]    {$0$};
    \draw (72,70.4) node [anchor=north west][inner sep=0.75pt]    {$0$};
    \draw (124,60.4) node [anchor=north west][inner sep=0.75pt]    {$0$};
    \draw (125,83.4) node [anchor=north west][inner sep=0.75pt]    {$0$};
    \draw (125,111.4) node [anchor=north west][inner sep=0.75pt]    {$0$};
    \draw (165,44.4) node [anchor=north west][inner sep=0.75pt]    {$0$};
    \draw (163,68.4) node [anchor=north west][inner sep=0.75pt]    {$0$};
    \draw (179,103.4) node [anchor=north west][inner sep=0.75pt]    {$0$};
    \draw (206,119.4) node [anchor=north west][inner sep=0.75pt]    {$0$};
    \draw (42,24.4) node [anchor=north west][inner sep=0.75pt]    {$1$};
    \draw (93,46.4) node [anchor=north west][inner sep=0.75pt]    {$1$};
    \draw (124,16.4) node [anchor=north west][inner sep=0.75pt]    {$1$};
    \draw (108,68.4) node [anchor=north west][inner sep=0.75pt]    {$1$};
    \draw (187,76.4) node [anchor=north west][inner sep=0.75pt]    {$1$};
    \draw (209,90.4) node [anchor=north west][inner sep=0.75pt]    {$1$};
    \draw (233,108.4) node [anchor=north west][inner sep=0.75pt]    {$1$};
    \draw (21,2.4) node [anchor=north west][inner sep=0.75pt]    {$T_{0}$};
    \draw (100,-0.6) node [anchor=north west][inner sep=0.75pt]    {$T_{1}$};
    \draw (166,0.4) node [anchor=north west][inner sep=0.75pt]    {$T_{2}$};

    \end{tikzpicture}

    \caption{Example binary AIFV-$3$ code for source alphabet $\left\{ a,  b,  c,  d\right\}.$ The small nodes are either complete or intermediate \ nodes, while the large nodes are master nodes with their assigned source symbols. Note that $T_2$ encodes $\hat a$, which is at its  root, with an empty string!}
    \label{fig:example_aifv_3X}
\end{figure}
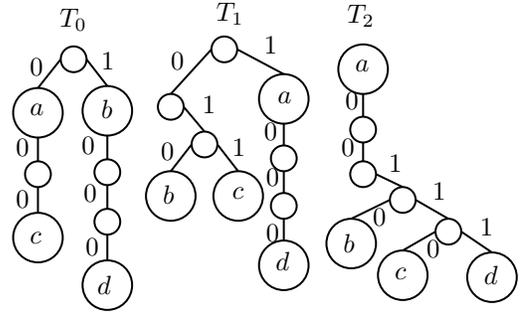
\end{definition}
Consequences of the Definitions:
\begin{enumerate}[label=\alph*)]
    \item Every leaf of a code tree must be a master node of degree 0. In particular, this implies that every code tree contains at least one master node of degree 0.
    \item Definition \ref{def: binary aifv-m codes}, and in particular Condition (2), result in unique decodability (proven in \cite{aifv_m}).
    \item  For $k  \neq 1$, the root of a $T_k$ tree is permitted to be a master node. If a root is a master node, the associated codeword is the empty string (Figure \ref{fig:example_encodingX})! The root of a $T_1$ tree cannot be a master node.
    \item The root of a $T_k$ tree may be a intermediate-0 node.
    \item For $k > 0$, every $T_k$ tree must contain at least one intermediate-1 node, the node $0^k$. A $T_0$ tree might not contain any intermediate-1 node.
For $k > 0$, the root of a $T_k$ tree cannot be a intermediate-1 node. The
root of a $T_0$ tree is permitted to be a intermediate-1 node.
\end{enumerate}
We now describe the encoding and decoding procedures. These are illustrated in Figures \ref{fig:example_encodingX} and \ref{fig:example_decodingX} which demonstrate the unique decodability
property of binary AIFV-m codes.\\
\newline
\textbf{Procedure 1}(Encoding of a Binary AIFV-m Code). A source sequence
$\alpha_1,\alpha_2 \cdots$is encoded as follows: Set $T = T_0$ and $i = 1$.
\begin{enumerate}[leftmargin=*]
\item Encode $\alpha_i$ using $T$
\item Let k be the index such that $\alpha_i$ is encoded using a degree-k master node in $T$
\item Set $T = T_k$; $i = i + 1$
\item Goto line 1
\end{enumerate}
\begin{figure}[htbp]
    \centering
    \begin{tikzpicture}[thick]
        \node (T1) at (0,0) {$T_0$};
        \node (T2) at (2,0) {$T_0$};
        \node (T3) at (4,0) {$T_2$};
        \node (T4) at (6,0) {$T_1$};
        \node (s1) at (0,1) {$ c$\vphantom{$ {abcd}$}};
        \node (s2) at (2,1) {$ b$\vphantom{$ {abcd}$}};
        \node (s3) at (4,1) {$ a$\vphantom{$ {abcd}$}};
        \node (s4) at (6,1) {$ b$\vphantom{$ {abcd}$}};
        \node (c1) at (0,-1) {$000$};
        \node (c2) at (2,-1) {$1$};
        \node (c3) at (4,-1) {$\epsilon$};
        \node (c4) at (6,-1) {$010$};
        \draw[->] (T1) -- (T2);
        \draw[->] (T2) -- (T3);
        \draw[->] (T3) -- (T4);
        \draw[->] (s1) -- (T1);
        \draw[->] (s2) -- (T2);
        \draw[->] (s3) -- (T3);
        \draw[->] (s4) -- (T4);
        \draw[->] (T1) -- (c1);
        \draw[->] (T2) -- (c2);
        \draw[->] (T3) -- (c3);
        \draw[->] (T4) -- (c4);
    \end{tikzpicture}
    \caption{Encoding $ c  b  a  b$ using the binary AIFV-$3$ code in Figure \ref{fig:example_aifv_3X}.
    The first $ c$ is encoded using a degree-$0$ master node (leaf) in $T_0$ so the first $ b$ is encoded using $T_0.$  This $ b$ is encoded using a degree-$2$ master node, so  $ a$ is encoded using $T_2.$ $ a$ is encoded using a degree-$1$ master node, so the second $ b$ is  encoded using $T_1\ldots.$}
    \label{fig:example_encodingX}
\end{figure}
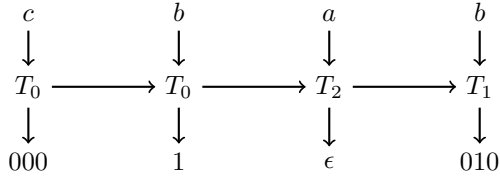

\textbf{Procedure 2} (Decoding of a Binary AIFV-m Code). Let $\beta$ be a binary
string that is the encoded message. Set $T = T_0$ and $i = 1$.
\begin{enumerate}[leftmargin=*]
\item Let $\beta_i$ be longest prefix of $\beta$ that corresponds to a path from the root of $T$
to some master node $M$ in $T$
\item Let $k$ be the degree of $M$ (as a master node) in $T$
\item Set $\alpha_i$ to be the source symbol assigned to $\beta_i$ in $T$
\item Remove $\beta_i$ from the start of $\beta$
\item Set $T = T_k$; $i = i + 1$
\item Goto line 1
\end{enumerate}

\begin{figure}[htbp]
    \centering
    \begin{tikzpicture}[thick]
        \node (T1) at (0,0) {$T_0$};
        \node (T2) at (2,0) {$T_0$};
        \node (T3) at (4,0) {$T_2$};
        \node (T4) at (6,0) {$T_1$};
        \node (c1) at (0,1) {$0001010$};
        \node (c2) at (2,1) {\sout{$000$}$1010$};
        \node (c3) at (4,1) {\sout{$0001$}$010$};
        \node (c4) at (6,1) {\sout{$0001$}$010$};
        \node (s1) at (0,-1) {$ c$\vphantom{$ {abcd}$}};
        \node (s2) at (2,-1) {$ b$\vphantom{$ {abcd}$}};
        \node (s3) at (4,-1) {$ a$\vphantom{$ {abcd}$}};
        \node (s4) at (6,-1) {$ b$\vphantom{$ {abcd}$}};
        \draw[->] (T1) -- (T2);
        \draw[->] (T2) -- (T3);
        \draw[->] (T3) -- (T4);
        \draw[->] (c1) -- (T1);
        \draw[->] (c2) -- (T2);
        \draw[->] (c3) -- (T3);
        \draw[->] (c4) -- (T4);
        \draw[->] (T1) -- (s1);
        \draw[->] (T2) -- (s2);
        \draw[->] (T3) -- (s3);
        \draw[->] (T4) -- (s4);
    \end{tikzpicture}
    \caption{Decoding $0001010$ using the binary AIFV-$3$ code in Figure \ref{fig:example_aifv_3X}.}
    \label{fig:example_decodingX}
\end{figure}
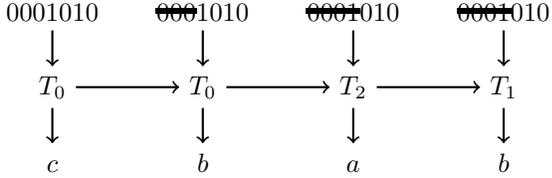

\begin{theorem}
    (\cite{aifv_m}, Theorem 3). Binary AIFV-m codes are uniquely decodable with delay at most $m$.
\end{theorem}
\begin{center}
    \textbf{ The cost of AIFV-m codes}
\end{center}

\begin{definition}
    Let $\mathcal{T}_k(m, n)$ denote the set of all possible type-k trees that
can appear in a binary AIFV-m code on $n$ source symbols. $T_k$ will be used
to denote a tree $T_k\in \mathcal{T}_k(m, n)$. Set
$\mathcal{T}(m,n) = \bigtimes_{i=0}^{m-1} \mathcal{T}_k(m, n)$.
\end{definition}
$\textbf{T} = (T_0, . . . , T_{m-1}) \in \mathcal{T}(m, n)$ will be a binary AIFV-m code.
\begin{definition}
Let $T_k \in \T_k(m, n)$ and $\sigma_i$ be a source symbol.
 \begin{itemize}
     \item $\ell(T_k, \sigma_i)$ denotes the length of the codeword in $T_k$ for $\sigma_i$.
     \item $d(T_k, \sigma_i)$ denotes the degree of the master node in $T_k$ assigned to $\sigma_i$.
     \item $\ell(T_k)$ denotes the average length of a codeword in $T_k$, i.e.,
     \begin{center}
         $\ell(T_k) =\sum_{i=1}^{n}
\ell(T_k, \sigma_i) \cdot p_i
$.
     \end{center}
     \item $M_j (T_k) = \{i \in \{1, 2, \cdots , n\} : d(T_k, \sigma_i) = j\}$ is the set of indices of source nodes that are assigned master nodes of degree $j$ in $T_k$. Set
     \begin{align*}
         \textbf{q}(T_k) = \big(q_0(T_k), \cdots , q_{m-1}(T_k)\big) \text{ where } \\ \forall j \in [m],q_j (T_k)= \sum_{i\in M_j(T_k)} pi.
         \end{align*}
 \end{itemize}
 $\sum_{j\in m}  q_j (T_k) = 1$
, so $\textbf{q}(T_k)$ is a probability distribution.\\
If a source symbol is encoded using a degree-j master node in $T_k$, then
the next source symbol will be encoded using code tree $T_j$. Since the source
is memoryless, the transition probability of encoding using code tree $T_j$ immediately after encoding using code tree $T_k$ is $q_j (T_k)$.\\

\begin{figure}[t]
        \centering
        \includegraphics[width=3in]{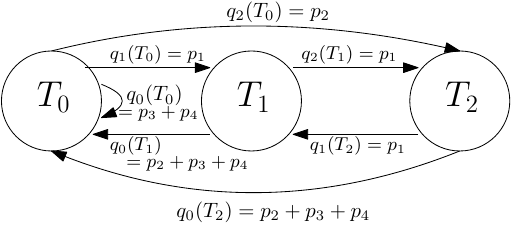}
        \caption{Markov chain corresponding to  AIFV-$3$ code in Figure \ref{fig:example_aifv_3X}. Note that $T_1$ contains no degree 1 master node, so there is no edge from $T_1$ to $T_1.$  Similarly, $T_2$ contains no degree 2 master node, so there is no edge from $T_2$ to $T_2.$}
        \label{fig:MarkovX}
        \end{figure}

This permits viewing the process as a Markov chain whose states are the
code trees. Figure \ref{fig:MarkovX} provides an example.\\
From Consequence (a) following Definition \ref{def: binary aifv-m codes}, $\forall k \in [m]$, every $T_k \in
\T_k(m, n)$ contains at least one leaf, so $q_0(T_k) > 0$. Thus, as described in
Section \ref{sec: Minimum Cost Markov Chain Problem} this implies that the associated Markov chain is a unichain whose
unique recurrence class contains $T_0$ and whose associated transition matrix
$Q$ has a unique stationary distribution $\boldsymbol{\pi}$.
\end{definition}
\begin{definition}
    Let $\textbf{T} = (T_0, \cdots , T_{m-1}) \in \T (m, n)$ be some AIFV-m code,
$\textbf{Q}(\textbf{T})$ be the transition matrix of the associated Markov chain and
\begin{center}
    $\boldsymbol{\pi}(\textbf{T}) = \big(\pi_0(\textbf{T}), \cdots , \pi_{m-1}(\textbf{T})\big)$
\end{center}
be $\textbf{Q}(\textbf{T})$’s associated unique stationary distribution. Then the average cost
of the code is the average length of an encoded symbol in the limit, i.e.,
\begin{center}
    $cost(\textbf{T}) = \sum_{k=0}^{m-1} \ell(T_k) \cdot \pi_k(\textbf{T})$.
\end{center}
\begin{definition}(The Binary AIFV-m Code problem). Construct a binary
AIFV-m code $T \in T (m, n)$ with minimum $cost(T)$, i.e.,
\begin{center}
    $cost(\textbf{T}) = \min\limits_{\textbf{T}^{\prime}\in \T (m,n)}cost(\textbf{T}^{\prime}).$
\end{center}
This problem is exactly the minimum-cost Markov Chain problem introduced in section Section \ref{sec: Minimum Cost Markov Chain Problem} with $\mbbS_k = \T_k(m, n)$. 
\end{definition}
\end{definition}

\eject

\begin{center}
    \textbf{Proofs of Lemmas and Corollaries}
\end{center}

\textbf{Proof of  Corollary \ref{cor:gcor}.}
With the exception of the uniqueness  the proof  follows almost directly from the definitions. The details are provided in Corollary 3.3. in \cite{golinaifv-m}.  Uniqueness follows from our Theorem \ref {prop: main proposition for uniqueness}.
\hfill \qed

\medskip

\textbf{Proof of Lemma \ref{lemma:continuous}.}
First define  set 
\begin{align*}
    g_1^{\prime}=\Big\{\x=(x_1,x_2,y) \in \mathbb{R}^3\mid  x_1\in[0,1] \\
      \text{ and } y = g_0\big((x_1,x_2)\big)=g_2\big((x_1,x_2)\big) \Big\}
\end{align*}
By definition, 
\begin{equation}
    g_1^{\prime}=\big\{E_1(x_1)\text{ }|\text{ }x_1\in[0,1]\big\}.
\end{equation}

    Since we know that $g_0$ and $g_2$ are piecewise linear, 
    $g_1^{\prime}$ is a compact set. 

    The proof that $E_1$ is continuous will  be by contradiction.
 %
  Suppose $E_1$ is not a continuous function. Then there exists $x_1\in[0,1]$ and $\epsilon\in\mathbb{R}^+$  along with a sequence of numbers $Z=\{z_1,z_2,...\}\subset [0,1]$ satisfying the following properties. The $z_i$ converge to $x_1$ but 
 $d\big(E_1(x_1),E_1(z_i)\big)\geq \epsilon$ for all $z_i\in Z$ where $d\big(E_1(x_1),E_1(z_i)\big)$ denotes the euclidean distance between $E_1(x_1)$ and $E_1(z_i)$. 
 
 Since $g_1^{\prime}$ is compact there must exist a subsequence of $\big\{E_1(z_1),E_1(z_2),...\big\}$ that converges to a point in $g_1^{\prime}$. Since the points in $Z$ converge to $x_1$  this subsequence must converge to a point in $g'_1$ whose first coordinate is $x_1$ 
 and the only point in $g_1$ that satisfies this property is $E_1(x_1)$. But this contradicts the fact that  $d\big(E_1(x_1),E_1(z_i)\big)\geq \epsilon$ for all $z_i\in Z$.  Thus $E_1$ is continuous.
 
   $g_0(x_1,x_2)$ and $g_2(x_1,x_2)$ are both piecewise linear functions. This immediately implies that $g'_1$ is composed of pairwise linear pieces and thus $E_1(x_1)$ is a piecewise linear function.
   \hfill \qed

\medskip

\textbf{Proof of Lemma \ref{lemma: binary search algorithm}.}
 We will first show how to find $x_1^{\prime}$ and then explain how finding $x_2^{\prime}$ can be done using a similar procedure.

 Suppose $l < r$, $g_0\big(E_1^{\prime}(l)\big)\leq g_1\big(E_1^{\prime}(l)\big)$, and $g_0\big(E_1^{\prime}(r)\big)\geq g_1\big(E_1^{\prime}(r)\big)$.
Since  $g_0$, $g_1$ 
are continuous,
 there exists $x_1\in[l,r]$ such that $g_0\big(E_1^{\prime}(x_1)\big)= g_1\big(E_1^{\prime}(x_1)\big)$. 
 Recall that from the definition of $E(x_1)$
  $g_0\big(E_1^{\prime}(x_1)\big)= g_2\big(E_1^{\prime}(x_1)\big).$
  Thus
  $$g_0\big(E_1^{\prime}(x_1)\big)
  =g_1\big(E_1^{\prime}(x_1)\big)
  =g_2\big(E_1^{\prime}(x_1)\big)
  $$
  so $E'(x_1)=(x^*_1,x^*_2)$ and,
  in particular, $x_1=x_1^{*}$.
  
From Lemma \ref{lemma:poncare miranda} we know that $g_0\big(E_1^{\prime}(0)\big)\leq g_1\big(E_1^{\prime}(0)\big)$ and $g_0\big(E_1^{\prime}(1)\big)\geq g_1\big(E_1^{\prime}(1)\big)$. Thus, after  initializing $l=0$ and $r=1,$ a binary  halving search  finds $x_1^{\prime}$ as close as needed to $x^{*}$. 
This is illustrated in Algorithm \ref{alg:iterative-binarysearch}.

The algorithm stops after $b^{\prime}$ iterations with  $r-l=2^{-b^{\prime}}$ and $x_1^{*}\in[l,r]$ so  $|x_1^{*}-x_1^{\prime}|\leq2^{-b^{\prime}}.$ Setting $l=x_1^{\prime}$ yields the  desired $x_1^{\prime}$.

Since  Algorithm  \ref{alg:iterative-binarysearch} runs for $b^{\prime}$ iterations and in each iteration calculating $E_1^{\prime}(\mbox{mid})$ needs $O(n^5\cdot b)$ operations from Lemma \ref{lemma:2d-binarysearch}, the  algorithm \ref{alg:iterative-binarysearch} has  total running time  of $O(n^5\cdot b \cdot b^{\prime})=O(n^5\cdot b^2)$. Note that  Lemma \ref{lemma:2d-binarysearch} can always be applied because 
``$\mbox{mid}$'' can always be written using $b^{\prime}$ bits.

We now see  that we can similarly find $x_2^{\prime}$ satisfying $|x_2^{*}-x_2^{\prime}|\leq2^{-b^{\prime}}$.  The argument is almost exactly the same as for $x'_1.$

Let $\lambda\in[0,1]$ be  fixed. 
 Now consider the intersection of a type-$0$ and type-$1$ hyperplanes with $h_{2,\lambda}$:
\begin{eqnarray}
    f_0((x_1,\lambda), S_0) &=& \ell(S_0) + q_1(S_0)\cdot x_1 +q_2(S_0)\cdot \lambda \nonumber\\
    f_1((x_1,\lambda), S_1) &=& \ell(S_1) + q_1(S_1)\cdot x_1 +q_2(S_1)\cdot \lambda - x_1 \nonumber \\
    &=& \ell(S_1) + \big(q_1(S_1)-1\big) + q_2(S_1)\cdot \lambda \nonumber \\
    &=& \ell(S_1)- \big(q_0(S_1)+q_2(S_1)\big)\cdot x_1 \nonumber \\
    &+& q_2(S_1)\cdot \lambda
\end{eqnarray}
Thus  the intersection of each type-0 hyperplane with $h_{2,\lambda}$ corresponds to a line with a non-negative slope and the intersection of each type-1 hyperplane with $h_{2,\lambda}$ corresponds to a line with a negative slope (because we know $q_0(S_i)$ is  always positive). 
Thus the intersection of $g_0$ and with $h_{2,\lambda}$ is the lower-envelope of non-negative slope lines and the intersection of $g_1$ and with $h_{2,\lambda}$ is the lower-envelope of negative slope lines. 

From this the argument follows exactly the same as for $x_1.$
and we can also  find $x_2^{\prime}$ satisfying $|x_2^{*}-x_2^{\prime}|\leq2^{-b^{\prime}}$ in $O(n^5\cdot b^2)$. \hfill\qed

\medskip


\textbf{Proof of Lemma \ref{lemma: x* lies on Ti'}.}
Set  $T_0^{*}$, $T_1^{*}$ and $T_2^{*}$ be the trees associated  
with $\x^{*}=(x_1^{*},x_2^{*},y^{*})$, the unique intersection of $g_0$, $g_1$ and $g_2$.

Following Definition \ref{hp def},
 the hyperplanes associated with  $T_0^{*}$, $T_1^{*}$ and $T_2^{*}$ have the  formulas
\begin{equation}
    y=\cost(T_0^{*})+q_1 (T_0^{*})\cdot x_1+q_2 (T_0^{*})\cdot x_2
\end{equation}
\begin{equation}
    y=\cost(T_1^{*})-\big(q_0 (T_1^{*})+q_2 (T_1 )\big)\cdot x_1+q_2 (T_1^{*}) \cdot x_2
\end{equation}
\begin{equation}
    y=\cost(T_2^{*})+q_1 (T_2^{*})\cdot x_1-\big(q_0 (T_2^{*})+q_1 (T_2^{*})\big)\cdot x_2
\end{equation}

This can be rewritten as 
\begin{equation} \label{eq:intersection formula}
-
\begin{bmatrix}
\cost(T_0^{*})\\
\cost(T_1^{*})\\
\cost(T_2^{*})
\end{bmatrix}
=
\begin{bmatrix}
1 & q_1 (T_0^{*}) & q_2 (T_0^{*})\\
1 & q_1 (T_1^{*})-1 & q_2 (T_1^{*})\\
1 & q_1 (T_2^{*}) & q_2 (T_2^{*})-1
\end{bmatrix}
\cdot
\begin{bmatrix}
-y^{*}\\
x_1^{*}\\
x_2^{*}
\end{bmatrix}
\end{equation}
We denote the matrix used in equation \ref{eq:intersection formula} by $M$. \cite{golinaifv-m} proves that $M$ is invertible. We therefore have:
\begin{equation}
    -M^{-1}
\begin{bmatrix}
\cost(T_0^{*})\\
\cost(T_1^{*})\\
\cost(T_2^{*})
\end{bmatrix}
=
\begin{bmatrix}
-y^{*}\\
x_1^{*}\\
x_2^{*}
\end{bmatrix}
\end{equation}
We also have that:
\begin{align*}
M &=
    \begin{bmatrix}
1 & q_1 (T_0^{*}) & q_2 (T_0^{*})\\
1 & q_1 (T_1^{*})-1 & q_2 (T_1^{*})\\
1 & q_1 (T_2^{*}) & q_2 (T_2^{*})-1
\end{bmatrix} \\
&=
\begin{bmatrix}
1 & q_1 (T_0^{*}) & q_2 (T_0^{*})\\
1 & -q_0 (T_1^{*})-q_2 (T_1^{*}) & q_2 (T_1^{*})\\
1 & q_1 (T_2^{*}) & -q_0 (T_2^{*})-q_1 (T_2^{*})
\end{bmatrix}
\end{align*}
We know that $M^{-1}=\frac{adj M}{det M}$. We also know that each element in $M$ can be written in the format of $k\cdot 2^{-b}$ where $|k|\leq 2^{b}$. Using the Leibniz formula for determinants we see that $det M$ can be written as $2^{-2b}\cdot k$ where k is an integer and $|k|\leq 6\cdot 2^{2b}$.\\
Now consider the fact that each element in $adj M$ is $\pm1$ times the determinant of a $2\times 2$ sub matrix of $M$. We can use this fact to show that each element in $adj M$ can be written as  $2^{-2b}\cdot k^{\prime}$ where $k^{\prime}$ is an integer. Finally using $M^{-1}=\frac{adj M}{det M}$, we can show that each element of $M^{-1}$ can be written as  $\frac{k^{\prime}}{k}$ where $k$ and $k^{\prime}$ are integers and $|k|\leq6\cdot 2^{2b}$. Therefore $x_1^*$ and $x_2^*$ can be written as
\begin{equation}
    \frac{k_0^{\prime}}{k_0}\cdot \cost(T_0^{*})
    +\frac{k_1^{\prime}}{k_1} \cdot \cost(T_1^{*})+\frac{k_2^{\prime}}{k_2}\cdot \cost(T_2^{*})
\end{equation}
where $k_0^{\prime}$, $k_1^{\prime}$, $k_2^{\prime}$, $k_0$, $k_1$, $k_2$  are integers and $|k_0|$,$|k_1|$,$|k_2|$$\leq 6 \cdot 2^{2b}$ and we know that we can write $\cost(T_i^{*})$ as $\cost_i^{*} \cdot 2^{-b}$ where $\cost_i^{*}$ is an integer. We can now write $x_1^*$ and $x_2^*$ in the following format:
\begin{equation}
    \frac{k_0^{\prime}\cdot \cost_0^{*}}{k_0\cdot 2^b}
    +\frac{k_1^{\prime}\cdot \cost_1^{*}}{k_1\cdot 2^b}+\frac{k_2^{\prime}\cdot \cost_2^{*}}{k_2\cdot 2^b}
\end{equation}
This  means $x_1^*$ and $x_2^*$ can be written as $\frac{a_1}{b_1\cdot 2^b}$ and $\frac{a_2}{b_2\cdot 2^b}$ where $b_1$, $b_2$ $\leq6^3\cdot2^{6b}\leq2^{6b+8}$. This is because all of the terms used for $x_i^*$ have a common denominator of $2^{b}\cdot k_0\cdot k_1 \cdot k_2$.
\\
\newline
We will now show by contradiction that $\x^{*}$ lies on the hyperplane associated to $T_0^{\prime}$. If $\x^{*}$ doesn't lie on the hyperplane associated to $T_0^{\prime}$ then we know that the hyperplane associated to $T_0^{\prime}$ must be higher than the hyperplane associated to $T_0^{*}$ at $(x_1^{*},x_2^{*})$. We therefore have:
\begin{align*}
    &\cost(T_0^{\prime} )+q_1 (T_0^{\prime} ) x_1^*+q_2 (T_0^{\prime} ) x_2^*>\\ & \cost(T_0^* )+q_1 (T_0^* ) x_1^*+q_2 (T_0^* ) x_2^*.
\end{align*}
Thus,
\begin{align*}
\cost(T_0^{\prime} )-\cost(T_0^*)&+\big(q_1 (T_0^{\prime} )-q_1 (T_0^*)\big)x_1^*\\
&+\big(q_2 (T_0^{\prime} )-q_2 (T_0^*)\big) x_2^*>0
\end{align*}
Therefore,
\begin{align} \label{eq:type0-inequality}
\cost(T_0^{\prime} )-\cost(T_0^* )&+\big(q_1 (T_0^{\prime} )-q_1 (T_0^* )\big)\frac{a_1}{b_1} 2^{-b} \nonumber \\ 
&+\big(q_2 (T_0^{\prime} )-q_2 (T_0^* )\big) \frac{a_2}{b_2} 2^{-b}>0.
\end{align}
Hence,
\begin{align*}
\cost(T_0^{\prime} )-\cost(T_0^* ) &+\big(q_1 (T_0^{\prime} )-q_1 (T_0^* )\big)\cdot x_1^*\\&+\big(q_2 (T_0^{\prime} )-q_2 (T_0^* )\big)\cdot x_2^*\geq 2^{-(14b+16)}
\end{align*}
The last inequality comes from the fact that all the terms in equation (\ref{eq:type0-inequality}) have a common denominator of $2^{2b}\cdot b1 \cdot b2$ which is at most $2^{(14b+16)}$.\\
\newline
Since $T_0^{\prime}$ is the lowest type-0 hyperplane at $(x_1^{\prime},x_2^{\prime})$ we have:
\begin{equation}
\cost(T_0^{\prime} )+q_1 (T_0^{\prime} )\cdot x_1^{\prime}+q_2 (T_0^{\prime} )\cdot x_2^{\prime}\leq \cost(T_0^* )+q_1 (T_0^* )\cdot x_1^{\prime}+q_2 (T_0^* )\cdot x_2^{\prime}.
\end{equation}
Thus,
\begin{align*}
\cost(T_0^{\prime} )-\cost(T_0^* )+\big(q_1 (T_0^{\prime} )-q_1 (T_0^* )\big)x_1^{\prime}+\big(q_2 (T_0^{\prime} )-q_2 (T_0^* )\big) x_2^{\prime}\\
\leq 0.
\end{align*}
Therefore,
\begin{align*}
    \cost(T_0^{\prime} )-\cost(T_0^* )+\big(q_1 (T_0^{\prime} )-q_1 (T_0^* )\big) x_1^*+\big(q_2 (T_0^{\prime} )-q_2 (T_0^* )\big) x_2^*\\
    \leq \big(q_1 (T_0^{\prime} )-q_1 (T_0^* )\big) (x_1^*-x_1^{\prime} )+\big(q_2 (T_0^{\prime} )-q_2 (T_0^* )\big)(x_2^*-x_2^{\prime} )
\end{align*}
\begin{eqnarray}
\Longrightarrow   2^{-(14b+16)} &\leq&\big(q_1 (T_0^{\prime} )-q_1 (T_0^* )\big)\cdot (x_1^*-x_1^{\prime} ) \nonumber \\
&+&\big(q_2 (T_0^{\prime} )-q_2 (T_0^* )\big)\cdot(x_2^*-x_2^{\prime} )\nonumber \\ &\leq&|x_1^*-x_1^{\prime} |+|x_2^*-x_2^{\prime}|\nonumber\\&\leq& 2\cdot 2^{-(14b+18)}\nonumber \\&=& 2^{-(14b+17)}
\end{eqnarray}
This contradiction shows that $\x^{*}$ must lie on on the hyperplane associated to $T_0^{\prime}$. We can similarly show the same is true for $T_1^{\prime}$ and $T_2^{\prime}$. \hfill \qed

\end{document}